\documentclass[11pt]{article}
\usepackage{xr}
\usepackage{tikz}
\usetikzlibrary{positioning}
  \usepackage{pgf}
  \usepackage{colortbl}
\usepackage{amsmath,amssymb}
\usepackage{float}
\usepackage{bbm}
\usepackage{xcolor}
\usepackage{enumerate}
\usepackage[authoryear]{natbib}
\usetikzlibrary{arrows,positioning,automata,calc,fit,shapes.geometric,arrows.meta}
\usepackage[a4paper,
            bindingoffset=0.2in,
            left=0.5in,
            right=0.5in,
            top=1.5in,
            bottom=1.5in,
            footskip=.5in]{geometry}

\usepackage{rotating,multirow,makecell,array}
\usepackage{caption}
\usepackage{subcaption}
 \usepackage{bbm}

\usepackage{subfiles}
\usepackage{scalerel,stackengine}
\usepackage{lipsum}
\stackMath
\newcommand\reallywidehat[1]{%
\savestack{\tmpbox}{\stretchto{%
  \scaleto{%
    \scalerel*[\widthof{\ensuremath{#1}}]{\kern-.6pt\bigwedge\kern-.6pt}%
    {\rule[-\textheight/2]{1ex}{\textheight}}
  }{\textheight}%
}{0.5ex}}%
\stackon[1pt]{#1}{\tmpbox}%
}
\usepackage[english]{babel}
\usepackage{amsthm}
\usepackage{bbm}
\usepackage{blindtext}
\usepackage{dsfont}
\usepackage{algorithm}
\usepackage{algpseudocode}
\usepackage{amsmath,amssymb}
\usepackage{multirow,makecell}

\usepackage[utf8x]{inputenc}

\usepackage{tabstackengine}

\makeatletter
\def\BState{\State\hskip-\ALG@thistlm}
\makeatother
\usepackage{amssymb}
\usepackage[inline]{enumitem}
\usepackage{url} 

\usepackage[pdftex,colorlinks]{hyperref}

  \hypersetup{
colorlinks =true,
citecolor    = blue,
citebordercolor = violet,
filebordercolor=red,
linkbordercolor=blue
}
\usepackage{authblk}
\usetikzlibrary{shapes.geometric}

\usetikzlibrary{shapes,snakes}
\makeatletter
\newcommand*{\addFileDependency}[1]{
  \typeout{(#1)}
  \@addtofilelist{#1}
  \IfFileExists{#1}{}{\typeout{No file #1.}}
}
\makeatother
 
\def\spacingset#1{\renewcommand{\baselinestretch}%
{#1}\small\normalsize} \spacingset{1}
\newcommand{\eqn}{\begin{eqnarray}}
\newcommand{\ee}{\end{eqnarray}}
\newcommand{\eqnn}{\begin{eqnarray*}}
\newcommand{\een}{\end{eqnarray*}}
\newcommand{\ea}{\end{align}}
\newcommand{\be}{\begin{eqnarray}}

\newcommand{\ba}{\begin{align}}

\def\bSig\mathbf{\Sigma}

\DeclareMathOperator{\logit}{logit}

\title{Semiparametrically Efficient Score for the Survival Odds Ratio}{}

\author
{
Denise Rava, drava@ucsd.edu \\
Department of Mathematics, University of California, San Diego
\and
Jelena Bradic, jbradic@ucsd.edu \\
Department of Mathematics and Halicioglu Data Science Institute\\
 University of California, San Diego

\and
Ronghui Xu, rxu@health.ucsd.edu \\
Department of Mathematics and Herbert Wertheim School of Public Health and Halicioglu Data Science Institute\\
 University of California, San Diego
}

\date{}

\begin{document}

\maketitle

\abstract{ We consider a  general proportional odds model for survival data under binary treatment, 
where the  functional form of the covariates is left unspecified.
We derive the efficient score for the conditional survival odds ratio given the covariates
 using modern semiparametric theory.
The efficient score may be useful in the development of doubly robust estimators, although computational challenges remain.  
 }

\newtheorem{lemma}{Lemma}
\newtheorem{theorem}{Theorem}
\newtheorem{assumption}{Assumption}

\maketitle

\section{Introduction}


 The commonly used hazard ratio recently come under scrutiny due to its conditioning on the  risk sets which became  not  exchangeable 
 even under randomized treatment  assignment at time zero 
\citep{hernan2010hazards,aalen2015does, 
martinussen2020subtleties}.
While the debate is still on-going \citep{prentice2022intention, ying:xu}, a useful alternative measure of  treatment effect is the  survival odds ratio.
The proportional odds model for survival data,   assuming  the  odds ratio to be constant over time given 
the treatment and the covariates, has been well studied in the literature 
\citep{bennett1983analysis,pettitt1984proportional,  dinse1983regression, 
dabrowska1988estimation,murphy1997maximum,rossini1996semiparametric,scharfstein1998semiparametric,yang1999semiparametric, xu2001semiparametric, 
collet2003modelling,crowder2017statistical}.
 It is a special case of the transformation model with error belonging to the $G^\rho$ family of \cite{harr:flem}, where $\rho=1$. 
It models attenuating hazard ratios among treatment groups and also provides a useful alternative to the otherwise commonly used proportional hazards assumption, the latter equivalent to the $G^\rho$ family of error distribution with $\rho=0$. 

In this paper we consider a more general proportional odds model for survival data, 
and relax the  functional forms of the covariates in the model.
We derive the efficient score for the conditional survival odds ratio 
 using modern semiparametric theory.
Efficient scores are useful in practice, not only because they are  efficient, but also because they can be useful in the development of doubly robust estimators \citep{bickel1993efficient,tsiatis2007semiparametric,zheng2016doubly,dukes2019doubly,rava2023doubly}.

\subsection{Model and notation}

Denote $T$ the time to failure, $C$ the censoring random variable and $X=\min{(T,C)}$ the observed (and possibly censored) failure time.
Let $\delta$ be the event indicator, $A=0,1$ a binary treatment, $Z$ be a vector of baseline covariates, and $\tau <\infty$ be an upper limit to follow-up time.  

Denote $S(t | A, Z)$ the conditional survival function of $T$ given $A$ and $Z$. 
Consider the following model for the distribution of $T$ given $A$ and $Z$: 
\be\label{model}
\logit \left\{S(t | A,Z)\right\}=\beta A-G(t,Z),
\ee
where $\logit(x)=\log\{x/(1-x)\}$ and $G(t,Z)$ is completely unspecified.
This is more general than the classical proportional odds model for survival data that assumes $G(t,Z)$ to be 
a baseline function in $t$ plus a linear function in $Z$ \citep{bennett1983analysis,pettitt1984proportional,chen2012estimating}. 
Under model \eqref{model} the treatment effecte expressed as
 the log survival odds ratio between the treated and the untreated, is captured by the estimand $\beta$:
 \be\label{effect}
\beta=\log\left\{\frac{S(t | A =1, Z)}{1-S(t | A = 1, Z)} \left/ \frac{S(t | A = 0, Z)}{1-S(t | A = 0, Z)} \right. \right\}.
\ee

We use the counting process and the at-risk process notation: $N(t)=\mathbbm{1}\left\{X\leq t, \; \delta=1\right\}$ and $Y(t)=\mathbbm{1}\left\{X\geq t\right\}$.
Define the martingale $M(t):=N(t)-\int_0^tY(u)\lambda(u | A,Z)du$, where $\lambda(t | A,Z)$ is the conditional hazard function. 
We make the standard assumption $C \perp T  |  A,Z$, where  $\perp$ indicates statistical independence.

\section{Semiparametrically efficient score}


In the following we derive the semiparametrically efficient score for the
 estimand $\beta$, following the approach described in \cite{tsiatis2007semiparametric}.
We consider the reparametrization as in \cite{chen2012estimating}: $R(t,Z)=\exp\{G(t,Z))\}$. 
 Define also
$r(t,Z)=\partial_t R(t,Z)$, where $ \partial_t$ is shorthand for $ \partial/ \partial_t$.
Under model \eqref{model}, it can be verified immediately that:
\be\label{S}
 S(t | A,Z)= \frac{\exp(\beta A)}{\exp(\beta A)+R(t,Z)},
\ee
and
\be \label{lambda}
\lambda(t | A,Z)=\frac{r(t,Z)}{\exp (\beta A)+  R(t,Z)}.
\ee


Under model \eqref{model}, the likelihood for a single copy of the data is:
\eqnn
L(\beta, R,\Lambda_c,f) &=&\left\{\frac{r(X,Z)}{\exp (\beta A)+ R(X,Z)}\right\}^\delta\frac{\exp(\beta A)}{\exp(\beta A)+R(X,Z)}
\\
&& \cdot \lambda_c(X | A,Z)^{1-\delta}\exp\{-\Lambda_c(X | A, Z)\}f(A,Z),
\een
where $\Lambda_c(t | a, z)=\int_0^t\lambda_c(u |a,z)du$, $\lambda_c(t | a,z)$ is the conditional hazard function for $C$ given $A$ and $Z$, and $f(a,z)$ is the joint density or probability function of $A$ and $Z$. Here the parameter of interest  is $\beta$ while the nuisance parameter is 
$\eta=\left[R,\lambda_c,f\right]^\top$.
Following lemma 2 of \citet{rava2023doubly}, 
 the score for $\beta$ is:
\be\label{Sbeta}
S_\beta=\frac{\partial\log L }{\partial\beta}=-\int_0^\tau A \cdot S(t | A,Z) dM(t).
\ee

Let $\mathbf\Lambda$ be the nuisance tangent space, i.e. the space spanned by the nuisance score function, and let $\mathbf\Lambda^\perp$ be its orthogonal complement.
The efficient score is $S_{eff}=S_\beta- \Pi\{S_\beta | \mathbf\Lambda\} = \Pi\{S_\beta | \mathbf\Lambda^\perp\}$  \citep{tsiatis2007semiparametric}, where $\Pi\{S_\beta | \mathbf\Lambda\}$ and $\Pi\{S_{\beta} | \mathbf\Lambda^\perp\}$ are the projections of $S_{\beta}$ onto $\mathbf\Lambda$ and $\mathbf\Lambda^{\perp}$, respectively.


Let $\Lambda_{1s}$,  $\Lambda_{2s}$ and $\Lambda_{3s}$ be the nuisance tangent spaces of $R$, $\lambda_c$ and $f$, respectively.
Denote $M_c(t)$ the martingale associated with the censoring distribution.
The following Lemmas describe $\mathbf\Lambda$ and $\mathbf\Lambda^\perp$ under model \eqref{model}.

\begin{lemma}\label{lemma1}
Under model \eqref{model}, the nuisance tangent space $\mathbf\Lambda$ is
\be\label{ds}
\mathbf\Lambda=\Lambda_{1s} \oplus \Lambda_{2s} \oplus \Lambda_{3s},
\ee
where
\be\label{lam1}
\Lambda_{1s}&=&\left\{\int_0^\tau \left\{\frac{\partial_t h(t, Z) }{\exp(\beta A)\lambda(t | A,Z)}-\frac{h(t,Z)}{\exp(\beta A)}\right\}S(t | A,Z)dM(t)\right.
\\
&&\left.\nonumber
\;\; : \; for\;\;all\;\;h(t,Z)\;\;s.t.\;\;\partial_t h(t, Z)\;\;exists\right\},
\ee
\be\label{lam2}
\Lambda_{2s}=\left\{\int_0^\tau \alpha (t, A, Z) dM_c(t)\; : \;for\;all\;\alpha(t,A,Z)\right\},
\ee
and 
\be\label{lam3}
\Lambda_{3s}=\left\{b(A,Z)\; : \;E\left[b(A,Z)\right]=0\right\}.
\ee
\end{lemma}


\begin{lemma}\label{lemma2}
Under model \eqref{model}, the orthogonal nuisance tangent space is:
\be
\mathbf\Lambda^\perp&=& \left\{\int_0^\tau\left\{h(t,A,Z)-\frac{\partial_t h_0(t, Z) }{\exp(\beta A)\lambda(t | A,Z)}+\frac{h_0(t,Z)}{\exp(\beta A)}\right\}S(t | A,Z)dM(t)\right.\nonumber
\\
&&\left.\; : \;for\;\;all\;\;h(t,A,Z)\right\},
\ee
where $h_0(t,Z)$, for almost every $t$, is a solution of the integro-differential equation 
\ba\label{diffeqc}
&\partial_t h_0(t,Z)-h_0(t,Z)m(t,A,Z)-q\{t,A,Z,h\} \nonumber\\
&+v(t,A,Z)\int_0^t\left\{\partial_u h_0(u, Z)w(u,A,Z)- h_0(u, Z)k(u,A,Z)\right\}du=0.
\end{align}
The coefficients $m(t,A,Z),q\{t,A,Z,h\},v(t,A,Z),w(t,A,Z),k(t,A,Z)$ are provided in  \ref{appendixA}.
\end{lemma}

The proofs of all results can be found in \ref{mainproof}.

Since $S_\beta=-\int_0^\tau A \cdot S(t | A,Z)dM(t)$, its projection on $\mathbf\Lambda^\perp$ is the element of $\mathbf\Lambda^\perp$ that corresponds to $h(t,A,Z)=-A$; that is:
\ba\label{scoreeff}
S_{eff}=\int_0^\tau\left\{-A-\frac{\partial_t h_0(t, Z) }{\exp(\beta A)\lambda(t | A,Z)}+\frac{h_0(t,Z)}{\exp(\beta A)}\right\}S(t | A,Z)dM(t)=0,
 \end{align}
 where $h_0(t,Z)$ is a solution to the differential equation \eqref{diffeqc} with $h(t,A,Z)=-A$.

\section{Final remarks}

Given $n$ independently and identically distributed (i.i.d.) observations, 
one can estimate the unknown nuisance parameters that are needed: $R(t, z)$, 
$S_c(t | a, z)= \exp\{ -\Lambda_c(t|a, z)\}$ 
and $\pi(z)= P(A =1  | z)$, either nonparametrically or using parametric or semiparametric working models.
A major challenge for 
 implementation of the efficient score  \eqref{scoreeff} 
is the solution of the complicated integro-differential equation \eqref{diffeqc} \citep{lakshmikantham1995theory}. 

Recall that $h_0(t,Z)= h_0(t,Z ; R, S_c, \pi)$ is a solution of the integro-differential equation 
\eqref{diffeqc}, with $h(t,A,Z)=-A$ for the efficient score  \eqref{scoreeff}. 
With the estimates $\hat R(t,Z),\hat S_c(t | A, Z)$ and $ \hat \pi(Z)$,   
denote $\hat h_0(t,Z) = h_0(t,Z_i; \hat R, \hat S_c, \hat \pi)$
 an estimator for $h_0(t,Z)$. 
The parameter of interest $\beta$ can then be estimated by solving
\ba
&\sum_{i=1}^n \delta_iV_i\left(X_i; \beta, \hat R, \hat S_c, \hat \pi\right)\frac{\exp(\beta A_i)}{\left\{\exp(\beta A_i)+ \hat R(X_i,Z_i)\right\}} \nonumber\\
&-\sum_{i=1}^n\int_0^{X_i} V_i\left(t; \beta, \hat R, \hat S_c, \hat \pi\right)\frac{\hat r(t,Z_i)\exp(\beta A_i)}{\left\{\exp(\beta A_i)+ \hat R(t,Z_i)\right\}^{2}}dt=0,
\end{align}
where
\ba
V_i\left(t; \beta, \hat R, \hat S_c, \hat \pi\right) 
&=-A_i-\frac{\partial_t \hat h_0(t,Z)\left\{\exp(\beta A_i)+ \hat R(t,Z_i)\right\}}{\exp(\beta A_i) d\hat R(t,Z_i)} 
\nonumber\\
&\:\;\;+\frac{\hat h_0(t,Z)}{\exp(\beta A_i)}. 
\end{align}

One may consider 
numerically approximating such solution for various values of $t$ on a grid of $[0,\max(X_1,\ldots,X_n)]$ and for every $Z_i$.
If a grid of $m$ elements is chosen this means numerically approximating the solution of $n*m$ integro-differential equations.
This can be therefore computationally intensive and choosing the appropriate $m$ and the appropriate grid may be non-trivial, and 
is  beyond the scope of this work.

For the asymptotic properties of this estimator we may first assume that the expression of $h_0(t,Z ; R, S_c, \pi)$ is known. Let $\hat \beta$ be  the solution of $\sum_{i=1}^n  U_i\left(X_i; \beta, \hat R, \hat S_c, \hat \pi, h_0\right)=0$.
The $\sqrt{n}$-consistency and asymptotic normality of $\hat \beta$ can be established using similar techniques as in \cite{hou2019estimating, luo2023doubly, rava2023doubly}.
Under certain regularity conditions, one can apply the mean value theorem to derive: 
\ba
\sqrt{n}(\hat{\beta}-\beta)=\left[\frac{1}{n}\sum_{i=1}^n  \frac{\partial}{\partial \beta}U_i\left(X_i; \tilde\beta,  \hat R,  \hat S_c, \hat \pi, h_0\right)\right]^{-1}\frac{1}{\sqrt{n}}\sum_{i=1}^n  U_i\left(X_i; \beta,  \hat R, \hat S_c, \hat \pi, h_0\right),\nonumber
\end{align}
where $\tilde \beta$ is  between $\hat \beta$ and the true value $\beta_0$.
Under certain convergence assumptions on the estimators of the nuisance parameters $\hat R, \hat S_c, \hat \pi$, one can use the orthogonality of the score to show that 
\ba
&\sqrt{n}(\hat{\beta}-\beta)\\\nonumber
&\;\;\;=\left[-\frac{1}{n}\sum_{i=1}^n  \frac{\partial}{\partial \beta}U_i\left(X_i; \tilde\beta,  R,  S_c, \pi, h_0\right)\right]^{-1}\frac{1}{\sqrt{n}}\sum_{i=1}^n  U_i\left(X_i; \beta,  R,  S_c, \pi, h_0\right) +o_p(1).\nonumber
\end{align}
The martingale central limit theorem can then be applied to $\sum_{i=1}^n  U_i\left(X_i; \beta,  R,  S_c, \pi, h_0\right)/ \sqrt{n}$ to conclude the proof of asymptotic normality.
If the estimator $\hat h_0$ is used instead of the true expression $h_0$ in the definition of $\hat \beta$, an extra layer of complications is added to the proof. Still we expect that under  assumptions on the rate of convergence of the estimator $\hat h_0$, the asymptotic normality of $\hat \beta$  holds.
The technical details of the proof are however beyond the scope of this work.

On the other hand, the efficient score might offer a starting point for the development 
of doubly robust estimators \citep{rava2023doubly}. 
This is would be of interest for future work.

\section*{Appendix}

\subsection*{Appendix A. Formulas}\label{appendixA}


We provide here the coefficients of the integro-differential equation introduced in Lemma \ref{lemma2}.

 \be
 m(t,A,Z)&=&\frac{r(t,Z)E\left[S_c(t \mid A,Z)\exp(\beta A)\left\{\exp(\beta A)+ R(t,Z)\right\}^{-3}\mid Z\right]}{E\left[S_c(t \mid A,Z)\exp(\beta A)\left\{\exp(\beta A)+ R(t,Z)\right\}^{-2} \mid  Z \right]},\nonumber\\
q\{t,A,Z,h(\cdot)\} 
&=&\frac{r(t,Z)E\left[h(t,A,Z)S_c(t \mid A,Z)\left\{\exp(\beta A)\right\}^2\left\{\exp(\beta A)+ R(t,Z)\right\}^{-3}   \mid  Z \right]}{E\left[S_c(t \mid A,Z)\exp(\beta A)\left\{\exp(\beta A)+ R(t,Z)\right\}^{-2}  \mid  Z \right]}\nonumber\\
&+&\frac{r(t,Z)\int_0^t r(u,Z)E\left[h(u,A,Z) S_c(u \mid A,Z) \left\{\exp(\beta A)\right\}^2\left\{\exp(\beta A)+ R(u,Z)\right\}^{-4}  \mid  Z \right]du}{E\left[S_c(t \mid A,Z)\exp(\beta A)\left\{\exp(\beta A)+ R(t,Z)\right\}^{-2}  \mid  Z \right]},\nonumber\\
v(t,A,Z)&=& \frac{r(t,Z)}{E\left[S_c(t \mid A,Z)\exp(\beta A)\left\{\exp(\beta A)+ R(t,Z)\right\}^{-2}  \mid  Z \right]},\nonumber\\
w(t,A,Z)&=&E\left[S_c(t \mid A,Z)\exp(\beta A)\left\{\exp(\beta A)+ R(t,Z)\right\}^{-3}\mid  Z \right],\nonumber\\
k(t,A,Z)&=& r(t,Z)E\left[S_c(t\mid A,Z)\exp(\beta A)\left\{\exp(\beta A)+ R(t,Z)\right\}^{-4} \mid  Z \right].\nonumber
 \ee

\subsection{Appendix B. Proofs of the main results}\label{mainproof}

We report here the proofs of the main results, Lemma \ref{lemma1} and Lemma \ref{lemma2}.

\begin{proof}[Proof of Lemma \ref{lemma1}]
The spaces $\Lambda_{1s}$,  $\Lambda_{2s}$ and $\Lambda_{3s}$ are orthogonal to each other because the corresponding nuisance parameters are variationally independent and the likelihood factors into the likelihoods for each of them \citep{tsiatis2007semiparametric}.
By this and by definition of $\mathbf\Lambda$, \eqref{ds} is directly proven.

Lemma 5.1 and pag.117 of \citet{tsiatis2007semiparametric} proves \eqref{lam2} and \eqref{lam3}, respectively.

We now prove \eqref{lam1}. When the nuisance parameter has finite dimension, the nuisance tangent space is simply the space spanned by the nuisance score. For semiparametric model, the nuisance tangent space is the mean-square closure of the nuisance tangent spaces of all parametric submodels. 
Let's consider the following generic parametric submodel: 
\be
S(t | A,Z ; \gamma)= \frac{\exp{\left(\beta A\right)}}{\exp{\left(\beta A\right)}+R(t,Z; \gamma)},
\ee
where $\gamma_0$ indicates the true value of the parameter.
If we focus on this parametric submodel, by Lemma 2 of \cite{rava2023doubly}, we have:
\be
S_{\gamma}&=&\left.\frac{\partial \log L}{\partial \gamma}\right|_{\gamma=\gamma_{0}}
\\
&=&\int_0^\tau \left\{\frac{ \partial_{\gamma}r(t,Z; \gamma_{0})}{ r(t,Z)}-\frac{ \partial_{\gamma} R(t,Z;\gamma_0)}{\exp (\beta A)+ R(t,Z)}\right\}dM(t).
\ee
We therefore conjecture the following:
\be\label{lam12}
\Lambda_{1s}&=&\left\{\int_0^\tau \left[\frac{ \partial_{t}h(t,Z)}{ r(t,Z)}-\frac{h(t,Z)}{\exp (\beta A)+ R(t,Z)}\right]dM(t)\right.
\\&&\left.\;\;:\;\;for\;\;all\;\;h(t,Z)\;\;s.t.\;\;\partial_t h(t, Z)\;\;exists\right\}.\nonumber
\ee
The above calculations prove that the nuisance tangent space of any parametric submodel belongs to $\Lambda_{1s}$ defined in \eqref{lam12}. 
To complete our proof we need to prove that, for any element of the conjectured $\Lambda_{1s}$, indexed by $h(t,Z)$, there exists a parametric submodel such that, such element is part of its nuisance tangent space. 
Given $h(t,Z)$, straightforward algebra proves that the score of the following parametric submodel
\be
S(t | A,Z)= \frac{\exp(\beta A)}{\exp(\beta A) + \gamma h(t,Z) + R(t,Z)}
\ee
corresponds to the element of $\Lambda_{1s}$ indexed by the chosen $h(t,Z)$.
Our conjecture is therefore proven right and so \eqref{lam12}.
Simplification of \eqref{lam12} leads to \eqref{lam1}.
\end{proof}

\begin{proof}[Proof of Lemma \ref{lemma2}]
If we don't put any restriction on the mechanism that generates the data, it follows from Theorem 4.4 of \citet{tsiatis2007semiparametric} that the tangent space is $\mathcal{H}=\{h(X,\delta,A,Z)\;:\;E(h)=0,\;E(h^\top h)<\infty\}$, i.e. the entire Hilbert space.
Since model \eqref{model} does not impose any restriction on the censoring distribution or the distribution of the treatment and the covariates, it is true that:
\be
\mathcal{H}=\Lambda^*_{1s} \oplus \Lambda_{2s} \oplus \Lambda_{3s},
\ee
where $\Lambda^*_{1s}$ is the tangent space associated with the distribution of $T$ given $A,Z$, now left arbitrary.
Therefore, by \eqref{ds}, $\mathbf\Lambda^\perp$ is the residual of the projection of an arbitrary element of $\Lambda_{1s}^*$ onto $\Lambda_{1s}$.
Similarly to the proof of \eqref{lam1}, it can be shown that:
\be
\Lambda_{1s}^*&=&\left\{\int_0^\tau h(t,A,Z)S(t | A,Z)dM(t)\;\;for\;\;all\;\;h(t,A,Z)\right\}.
\ee
We hence now focus on deriving $\prod\left[\int_0^\tau h(t,A,Z)S(t | A,Z)dM(t) | \Lambda_{1s}\right]$ for a generic $h(t,A,Z)$.

By definition of projection and by \eqref{lam1}, we need to find $h_0(t,Z)$ such that, for any $h(t,Z)$:
\be
0&=&E\left(\left[\int_0^\tau \left\{h(t,A,Z)-\frac{\partial_t h_0(t, Z) }{\exp(\beta A)\lambda(t | A,Z)}+\frac{h_0(t,Z)}{\exp(\beta A)}\right\}S(t | A,Z)dM(t)\right]\nonumber
\right.\\
&&\label{ortcond}\;\;\;\;
\left.\cdot\left[\int_0^\tau \left\{\frac{\partial_t h(t, Z) }{\exp(\beta A)\lambda(t | A,Z)}-\frac{h(t,Z)}{\exp(\beta A)}\right\}S(t | A,Z)dM(t)\right]\right).
\ee
Given two martingales $M,M'$, we use $<M,M'>$ to indicate the predictable covariation process.
By \eqref{ortcond} we have:
\begin{align*}
0&=E\left[\int_0^\tau \left\{h(t,A,Z)-\frac{\partial_t h_0(t, Z) }{\exp(\beta A)\lambda(t | A,Z)}+\frac{h_0(t,Z)}{\exp(\beta A)}\right\} \right.
\\& \left.\cdot \left\{\frac{\partial_t h(t, Z) }{\exp(\beta A)\lambda(t | A,Z)}-\frac{h(t,Z)}{\exp(\beta A)}\right\}S^2(t | A,Z)<dM(t)>\right].
\end{align*}
By \eqref{S}, \eqref{lambda} and some tedious algebra, we therefore have:
\begin{align*}
0&=E\left[\int_0^\tau\left\{h(t,A,Z)-\frac{\partial_t h_0(t, Z) }{\exp(\beta A)\lambda(t | A,Z)}+\frac{h_0(t,Z)}{\exp(\beta A)}\right\} \frac{\exp(\beta A)\partial_t h(t,Z)Y(t)}{\left\{\exp(\beta A) + R(t,Z)\right\}^2}dt  \right.
\\&\;\;\;\;\;\;\;\left.-\int_0^\tau \left\{h(t,A,Z)-\frac{\partial_t h_0(t, Z) }{\exp(\beta A)\lambda(t | A,Z)}+\frac{h_0(t,Z)}{\exp(\beta A)}\right\}\frac{\exp(\beta A)r(t,Z)h(t,Z)Y(t)}{\left\{\exp(\beta A) + R(t,Z)\right\}^{3}}dt\right].
\end{align*}
Applying integration by part on the second term gives us the following:
\ba\label{ortogon}
0&=E\left\{ \int_0^\tau H(t,A,Z)\partial_th(t,Z)dt\right\},
\end{align}
where
\ba
&H(t,A,Z) \nonumber
\\&=\left\{h(t,A,Z)-\frac{\partial_t h_0(t, Z) }{\exp(\beta A)\lambda(t | A,Z)}+\frac{h_0(t,Z)}{\exp(\beta A)}\right\} \frac{\exp(\beta A)Y(t)}{\left\{\exp(\beta A) + R(t,Z)\right\}^{2}}  \nonumber
\\&\;\;\;+\int_0^t\left\{h(u,A,Z)-\frac{\partial_u h_0(u, Z) }{\exp(\beta A)\lambda(u | A,Z)}+\frac{h_0(u,Z)}{\exp(\beta A)}\right\} \frac{\exp(\beta A)r(u,Z)Y(u)}{\left\{\exp(\beta A) + R(u,Z)\right\}^{3}}du. \nonumber
\end{align}

By \eqref{ortogon} we have that, for every $h(t,Z)$,
\ba
\int_0^\tau E\left[E\left\{H(t,A,Z) | Z\right\}\partial_th(t,Z)\right]dt=0.
\end{align}
This implies that, for almost every $t$,
\be\label{con}
E\left\{H(t,A,Z) | Z\right\}=0.
\ee 
The above can be easily proved by contraddiction setting $\partial_th(t,Z)=E\left\{H(t,A,Z) | Z\right\}$.
Condition \eqref{con} together with Lemma \eqref{alg} in \ref{appendixB} completes the proof.
\end{proof}

\subsection*{Appendix C. Additional lemmas}\label{appendixB}

We report here an additional lemma used in the proofs of the main results.

\begin{lemma}\label{alg}
For a generic functions $B(t,A,Z)$, we have:
\ba\label{alg1c}
E \left\{ B(t,A,Z)Y(t)\mid Z\right\}=E\left\{B(t,A,Z)S(t \mid A,Z) S_c(t \mid A,Z) \mid Z\right\}
\end{align}
\end{lemma}

\begin{proof}
By the tower law of conditional expectation and by the fact that $T \perp C \mid A,Z$, we have:
\ba
E \left\{ B(t,A,Z)Y(t)\mid Z\right\}&=E \left[ B(t,A,Z) E \left\{Y(t) \mid A,Z\right\}\mid Z\right]\\
&=E \left[ B(t,A,Z) S(t \mid A,Z) S_c(t \mid A,Z)\mid Z\right]
\end{align}
\end{proof}

\bibliographystyle{natbib}
\bibliography{causalpropoddsmodel}


\end{document}